\documentclass[12pt]{article}
\usepackage[colorlinks,pagebackref=false]{hyperref}
\hypersetup{citecolor=blue, linkcolor=blue, urlcolor = blue}

\usepackage{amsfonts,amsmath,amsthm,amssymb,verbatim,bm}
\usepackage{graphicx}
\usepackage{xcolor}
\usepackage{natbib}

\usepackage{setspace}
\usepackage{epsf}
\usepackage{fullpage}
\usepackage{bbm}
\usepackage{enumerate}
\usepackage{paralist}
\usepackage[margin=10pt,labelfont=bf,labelsep=endash]{caption}
\usepackage{subcaption}
\usepackage{floatrow}
\usepackage{multirow}

\usepackage{tikz}
\usetikzlibrary{arrows,decorations,decorations.pathreplacing,calc,matrix}
\definecolor{Red}{rgb}{1,0,0}
\definecolor{Blue}{rgb}{0,0,1}
\definecolor{Green}{rgb}{0,1,0}
\definecolor{magenta}{rgb}{1,0,.6}
\definecolor{lightblue}{rgb}{0,.5,1}
\definecolor{lightpurple}{rgb}{.6,.4,1}
\definecolor{gold}{rgb}{.6,.5,0}
\definecolor{orange}{rgb}{1,0.4,0}
\definecolor{hotpink}{rgb}{1,0,0.5}
\definecolor{newcolor2}{rgb}{.5,.3,.5}
\definecolor{newcolor}{rgb}{0,.3,1}
\definecolor{newcolor3}{rgb}{1,0,.35}
\definecolor{darkgreen1}{rgb}{0, .35, 0}
\definecolor{darkgreen}{rgb}{0, .6, 0}
\definecolor{darkred}{rgb}{.75,0,0}

\newcounter{desccount}

\newcommand{\dref}[1]{\hyperref[#1]{#1}}

\newcommand{\set}[1]{\{#1\}}

\newcommand{\bea}{\begin{align*}}
\newcommand{\eea}{\end{align*}}

\renewcommand{\bar}{\overline}
\renewcommand{\paragraph}[1]{\medskip \noindent \textit{#1}}

\def\s{\sigma}
\def\t{\theta}

\newcommand{\bk}{\color{black}}

\newtheorem{corollary}{Corollary}

\newtheorem{lemma}{Lemma}

\newtheorem{proposition}{Proposition}

\newtheorem{definition}{Definition}

\title{\sc{State-Promoted Investment for Industrial Reforms: an Information Design Approach}\footnote{We would like to thank seminar participants at Yonsei University. This research is supported by the Yonsei Signature Research Cluster Program of 2021 (No. 2021-22-0011, Myungkyu Shim) and by the National Natural Science Foundation of China under Grant No. 72003102 (Ji Zhang). The views expressed in this paper are those of the author and do not necessarily represent that of Korea Development Institute (KDI).}}

\author{
	\sc{Keeyoung Rhee}\thanks{Korea Development Institute (KDI). E-mail: ky.rhee829@gmail.com}
	\and
	\sc{Myungkyu Shim}\thanks{Yonsei University. E-mail: myungkyushim@yonsei.ac.kr}\\
	\and
	\sc{Ji Zhang}\thanks{PBC School of Finance, Tsinghua University. Email: zhangji@pbcsf.tsinghua.edu.cn}\\
}

\date{\today}
\begin{document}

\maketitle

\begin{abstract}

We analyze the optimal strategy for a government to promote large-scale investment projects under information frictions. Specifically, we propose a model where the government collects information on profitability of each investment and discloses it to private investors \emph{\`{a} la} \citet{KG2011}. We derive the government's optimal information policy, which is characterized as threshold values for the unknown profitability of the projects released to the private investors, and study how the underlying features of the economy affect the optimal policies. We find that when multiple projects are available, the government promotes the project with a bigger spillover effect by fully revealing the true state of the economy only when its profitability is substantially high. Moreover, the development of the financial/information market also affects the optimal rule.   \\

%Building on the Bayesian Persuasion framework, pioneered by \cite{KG2011}, we provide a theory to understand (1) how the government, which is assumed to have superior information about the project, can design a policy to finance the project, (2) effect of increased access by investors to better information on the design of the policy, and (3) possible capital (mis-)allocation problem in the presence of misalignment in the utility between the two agents.     \\

\end{abstract}
\vspace{-0.2in}

\thispagestyle{empty} % no page number \clearpage

\newpage

\pagebreak \setcounter{page}{1}
\thispagestyle{empty}

\section{Introduction} \label{intro}

State-promoted investments in developing economies, industrial policies to reform the industry structure for instance, have been analyzed in many dimensions. \cite{BMS2013} analyzed how a well-constructed industrial policies that relax financial constraints of selected firms can promote growth at the early stage of growth, while \cite{HK2009} quantified the misallocation problem possibly originated from government's policy. The government interventions in advanced economies draw great attentions of researchers as well: \cite{BG2015}, for example, analyzed when the government intervention, in the form of bailouts, in firms' activity can be effective and \cite{ALP2018} studied the policy implication in the economy with startup firms.

While the previous literature has examined the possible merits and demerits of such policies in the model with frictions, financial frictions as an example, there has been less attention to two important features of state-promoted investments that (i) there exists a gap between the information that can be acquired by the individual investors, who provide the fund, and the information the government has and (ii) such investment might exhibit positive externality on the aggregate economy. %less is known about designing optimal \textit{information} policy.
More specifically, first, the success rate and potential benefits of the investment are unknown ex ante. In particular, the benefits of certain investment require a long time to reveal. For instance, it took more than 30 years for Shenzhen in China to become a successful industrial city. It took more than a decade to enjoy economic/social profits from ``Gyeongbu Expressway,'' a massive investment on infrastructure connecting Seould and Busan, two major cities in South Korea. Each individual investor might have limited access to information to make precise prediction on profits from such industrial reforms, especially when the financial/information market is underdeveloped as was in South Korea in 1960 or in China in early 1990s. In contrast, compared to private investors (both domestic and foreign), the government is more aware of the success rate of the investment from government-run research institutes and/or experts.

	The second important feature of state-promoted investments for industrial reforms is that these investment projects exhibit potential positive externality on the aggregate economy as it is usually a massive investment on infrastructure and/or basic industries. In the competitive equilibrium, however, this sort of externality is not taken into account when investors make their own investment decisions. For instance, positive externality of the industrial transformation to reduce carbon dioxide (CO2) on the aggregate economy is not easy to measure from the perspective of each individual investor. Hence, if the government/policy maker has better information to predict future benefits and internalize additional benefits from the investment, it might be willing to pursue the investment project that will not be funded otherwise. %In this sense, it is natural to consider a tool that designs information to implement a large-scale investment in such an economy with a huge gap in information between the government and individual investors.

 This paper aims to analyze government's optimal strategy in providing information related to the return of a project that the government has the incentive to promote and facilitate its financing. In particular, we address the following questions: What would be the optimal information design when the investors do or do not have access to private information to evaluate the profitability of the investment projects? Under what conditions can the government increase the probability of initiating the investment project that benefits the economy with externality, which is not considered by the investor? Does greater transparency on the profitability of the investment project increase the probability to be successfully funded?

%In order to motivate our main idea, consider a government that would like to initiate a massive investment project in South Korea or China in 1970s and that in 2010s. Information-wise, there are two important distinctions between the two episodes; first, compared to private investors (both domestic and foreign), the governments in 1970s would have superior information about the success probability (or importance) of the project when compared to the ones in the 2000s. Second, private investors in 2010s have more access to (alternative) information on the investment projects than those in 1970s. These informational differences raise natural questions: Was it more easy for the government in 1970s to promote investments? Under what conditions can the government increase the probability of initiating the investment project? When does promoting the projects raise social welfare?

	In order to answer the above questions, we develop and analyze a model in a Bayesian Persuasion framework pioneered by \cite{KG2011}. We introduce the role of government in designing an information policy to finance a large investment project. Importantly, we assume that the return from the investment is uncertain; the government can run an experiment to learn the return and provide (an imperfect) signal about the return to the public. Once private investors update their own information sets, they can either (1) follow the government's signal or (2) conduct their own experiments to learn about the true profitability of the project. We further introduce the possibility that there might exist an extra benefit for the government from the successful investment, by paramaterizing the benefit that increases the government's utility while does not have any impact on the investor's utility: This reduced-form approach is to capture the idea that % a ruling party might increase their chance of winning the upcoming elections through the increased GDP.
individual investors would not internalize the benefit from the large scale investment that can potentially have positive externality/spillover effect on other industries.\footnote{While we focus on the positive side of the government, it is also possible to consider negative side; a ruling party might increase their chance of winning the upcoming elections through the increased GDP. This is a case with negative externality of the investment project: Our framework can be still applied to such a circumstance but normative implication would be different.} Finally, we also consider the equilibrium property of the model when there are two investment projects while only one of them can be funded in order to study the possible capital (mis-)allocation problem in our model.

There are several noteworthy findings from our analysis. First, the government sets up a threshold information disclosure rule to increase the probability of the project being funded. To fix the idea, we consider a benchmark case with a single project: When we design a signal structure such that the government manipulates information only when the benefit of the project is low, the optimal information policy is to lower transparency of the signal as far as possible.

	Second, when the cost of accessing to the additional information sources is substantially high so that investors should rely on the signal sent by the government, the government can set lower threshold for the signal. This reflects the fact that it is easier for the government to launch the project when the financial market is underdeveloped. %This distortion from the government becomes less severe as the cost to obtain private information (from the financial market) becomes lower, which is natural.
In other words, transparency of the government should become greater as more (private) information becomes accessible to each investor.

	Third, when there exist two projects with different degree of positive spillover/externality, the threshold level for each project crucially depends on its intrinsic profitability. Interestingly, the transparency level required to successfully launch a project with high (\textit{resp.} low) externality is monotonically (\textit{resp.} non-monotonically) related to its unknown profitability. In particular, there are two thresholds for the profitability, low threshold ($\underline{p}$) and high threshold ($\bar{p}$), in which the government would fully reveal its own information on its preferred project only when the profitability of the project is higher than $\bar{p}$. For less-preferred project, the government perfectly unveils its own information when the profitability is at the middle ($p \in (\underline{p},\bar{p})$) and not if profitability is either high or low.
	
	Fourth, as the positive effect of the preferred project becomes greater, the region at which the government would choose to be transparent becomes narrower. This is because the government has an incentive to launch the project regardless of its profitability. Lastly, the development of the financial market has a nontrivial effect on the thresholds for profitability; both threshold levels increase when the initial information cost is substantially low. This narrows the region at which transparency of the signal for the preferred project while has an ambiguous effect on the signal for the less-preferred one. When the initial cost is sufficiently high (but not too much), in contrast, low (\textit{resp.} high) threshold level increases (\textit{resp.} decreases). In this case, the region at which signal for preferred (\textit{resp.} non-preferred) project becomes more transparent enlarges (\textit{resp.} compresses).
%Lastly, we show further that the misallocation problem arises as an equilibrium outcome in the presence of the government's incentive to achieve its own profit maximization problem and imperfect information.

\section{Related Literature}	\label{sec:rl}

	Our work is related to several different strands of literature. First, as is discussed above, this paper is in line with the literature on industrial policy \citep{HK2009, BMS2013, BG2015, ALP2018}. Our research is differentiated from the previous literature by designing optimal information policy, which reflects the possible gap in information between investor and the government.

	Second, our paper is also related to the literature on Bayesian persuasion and information design pioneered by \citep{KG2011, che2013pandering}. Particularly, a number of recent studies focus on the information design problem in which multiple senders competitively provide information to influence a receiver's choice in their own favors \citep{hirsch2015competitive,gentzkow2016competition, boleslavsky2018limited}. A key assumption on their models is that each sender provides information on only one option in favor of the sender. However, in our model, there is only one sender who provides the receiver with additional information on more than one options, even including the option that the sender wishes not to be chosen. In this setup, we can study how a sender endogenously controls the quality of information on different options to encourage the receiver to select the option that the sender prefers the most, which cannot be captured in the previous literature.\footnote{\cite{dworczak2019simple} also provides an example of a multidimensional information design problem. Specifically, their approach allows the signal on one option to be correlated with the true state of the other option. However, we consider different situations in which a policy maker independently evaluates whether a project is worth being funded, which is believed to be more suited to state-promoted investment policies in many countries.} Furthermore, the recent work of \citet{candogan2021optimal} characterizes an optimal disclosure strategy when the receiver has private information. Our work focuses on the cases in which the receiver can access to her own private information \emph{ex post}, and therefore, the sender must take into account the receiver's incentive to acquire her private information in the information design problem.

	Last but not least, our finding also contributes to the literature on transparency of the policy makers on the information it possess. Since \cite{MS02}'s seminal work, there has been intensive debate on whether a policy maker should transparently release its own information or not (see \cite{AngeletosPavan:04}; \cite{AngeletosPavan:07}; \cite{amador_weill}; \cite{CFP14}; and \cite{GY2019} among many others). Our work considers a framework in which a government (policy maker) has an investment project that it would like to pursue, which provides a new perspective to this literature.

Remaining sections are organized as follows. Section \ref{sec:model} introduces the benchmark model. Section \ref{sec:bm} then provides the equilibrium properties of the benchmark model and Section \ref{sec:main} extends the benchmark analysis. Section \ref{sec:conclusion} concludes. All proofs are relegated to Appendix. %\ref{app_proof}.

	%In this note, I present a simpler model with binary return structure of the financial projects in an attempt to incorporate the private investor's incentive to acquire information for her own sake. As will be shown in the sequel, there are a number of interesting analytical results relative to the original model in which the project return is a continuous random variable but the signals are assumed to be a finite random variable.

\section{The Model}	\label{sec:model}

	We consider a two-date ($t = 0, 1$) economy of a policy maker (henceforth abbreviated to PM) and an investor at a financial market. When the game begins at $t = 0$, PM designs and conducts experiments on multiple financial projects respectively and then releases the experiment results to the investor. After receiving the information, the investor decides whether or not to gather more information on each financial project by herself, and then chooses one of these projects to fund. The return of the project funded by the investor is realized in $t = 1$.
	
	There are two projects in this economy, labeled $A$ and $B$. Each project, if funded by the investor, yields a binary stochastic return in $t = 1$. Specifically, we assume that the returns of projects $A$ and $B$ are \emph{i.i.d.} random variables $\t_A$ and $\t_B$ such that $\t_i = 1$ with probability $p \in (0, 1)$ and $\t_i = 0$ with probability $1 - p$ in $t = 1$ for each $i \in \set{A, B}$. Each project $i \in \set{A, B}$ requires a financial management cost $c > 0$ to be matured in $t = 1$.

 Importantly, we assume that two projects are heterogeneous in their positive effects on the aggregate economy. One possible interpretation of this effect is positive externality from industries that generate spillover effects on other industries. Each individual investor does not internalize such an effect so that he/she might undervalue the profit of the investment. On the contrary, the government, which might act as a social planner\footnote{In this sense, we consider a government in line with \cite{AJR05}, which is an institution that can spur the economic growth in the long-run.}, internalizes it to pursue the investment project. One might alternatively consider this situation as PM possessing superior information about the additional surplus from the project, denoted as $\mu_i > 0$, that is generated once project $i$ is funded. Individuals do not have enough information on it hence ignore the unobservable returns when making an investment decision. We assume that the spillover effect from the investment is sufficiently high, so PM prefers either project being funded:
	\begin{align}	\label{eq:invest}
	\mu_A > \mu_B > c.
	\end{align}
	
%	This additional surplus captures economic spillover effects of the investments in industries that cannot be accrued as financial return.

Without loss of generality, we assume that $\mu_A > \mu_B$; under the same condition, PM prefers project $A$ to be funded. Under this assumption, project $A$ refers to a new industry that a government wishes to replace with an old industry (project $B$). However, PM is assumed to own zero initial endowment, so she cannot self-finance any project.
	
	The investor is endowed with initial wealth $c > 0$, so she can finance either project $A$ or $B$, but not both. However, we assume
	\begin{align} \label{eq:risk}
	p < c.
	\end{align}

	Under this assumption, the investor does not want to finance any of these projects without any additional information because it is too risky to fund the projects.
	
	To promote investments, PM conducts a public experiment in a Bayesian persuasion framework \emph{\`{a} la} \citet{KG2011} to reduce uncertainty about the financial projects in need of funding. Specifically, PM commits to design a binary signal $s_i \in \set{0, 1}$ for each project $i \in \set{A, B}$ such that $Pr(s_i = 1 | \t_i = 1) = \s^1_i \in [0, 1]$ and $Pr(s_i = 0 | \t_i = 0) = \s^0_i \in [0, 1]$, and then disclose the realized signals to the investor. Without loss of generality, we throughout restrict our attention to the signal spaces $\s^1_A = \s^1_B = 1$ (PM always sends a signal that correctly reveals the profitability of the project if the project is really a profitable one) and simplify the notation $\s^0_i$ to $\s_i$.\footnote{To influence the investor's decision, any signal $\hat{s}_i$ is required to yield the posterior belief
$$\hat{p} := Pr(\t_i = 1 | \hat{s}_i = 1) > c > Pr(\t_i = 1 | \hat{s}_i = 0) = (1 - \hat{p}).$$
	For any feasible value of $\hat{p}$, it is not difficult to find a signal $s_i$ with $\s^1_i = 1$ and $\s^0_i \in (0, 1)$ that replicates the same posterior belief as $\hat{s}_i$.} Note that this conditional probability $\s_i$ sufficiently captures accuracy of the signal $s_i$. If $\s_i=1$, the signal from PM perfectly reveals the state of the corresponding project $i$; one can interpret this as ``perfect'' transparency case. On the contrary, there is an uncertainty on the state of the economy from the perspective of the investor if $\s_i \in (0,1)$ and this uncertainty becomes greater as $\s_i$ becomes lower. Hence, we can interpret transparency of the signal from PM as an increasing function of $\s_i$.

	We further assume that the investor can acquire additional information (signal) by paying extra cost. Specifically, we assume that the investor will perfectly know the true value of $\t_i$ if she exerts a non-financial effort $e > 0$ for each project $i \in \set{A, B}$. The effort cost $e > 0$ represents the status of financial market development. That is, market participants can find useful investment-related information more effectively as the financial market operates more efficiently. For example, we can correspond less-developed (\emph{resp.} well-developed) financial market in South Korea in 1960s (\emph{resp.} 2010s) to the case of a high $e$ (\emph{resp.} a low $e$). Namely, in the economy with underdeveloped financial market, individual investors cannot obtain valuable information on the profitability hence they should rely on public information disseminated by PM. In contrast, once the financial market (or information market) develops further, access to information becomes relatively easy, which will lower the investors' reliance on the public information.

	Lastly, for simplicity of our analysis, we assume that the investor, when she is indifferent between two projects, breaks a tie in favor of funding project $A$.\footnote{In fact, the assumption $\mu_A > \mu_B > 0$ implies that if the investor randomizes in project selection in the event of an indifference, PM can strictly increase his expected payoff by changing the signal structure to slightly increase the posterior success probability of project $A$.}
		
\section{Benchmark Analysis: Economy with Single Project}	\label{sec:bm}

We first study a benchmark economy in which there exists only one project, project $A$. This benchmark analysis will facilitate to understand the basic framework of strategic interplay between PM and the investor and its impact on informativeness of PM's signal. It will also provide a context for our main analysis that will follow in Section \ref{sec:main}.

It would be convenient for the expositional purpose to define the equilibrium at first:

\begin{definition} [Equilibrium of single-project economy]
An equilibrium of the economy with single-project is characterized by an information policy $\s_A \in (0, 1]$, such that:

 (1) It is optimal for the investor not to exert an effort for the additional information;

 (2) based on the posterior belief updated after PM's signal, the investor funds one project that gives the highest positive expected payoff;

 (3) the information policy maximizes PM's expected payoff.

\end{definition}
	
	We first study how informative PM's signal $s_A$ should be in his best favor by persuading the investor to fund the project. In other words, finding a cut-off value for the signal to satisfy the first part of the definition. Since $p < c$ from \eqref{eq:risk}, the investor will not fund the project without any additional information. Since we focus on $Pr(s_A = 1 | \t_A = 1) = 1 \ge 1 - \s_A = Pr(s_A = 1 | \t_A = 0)$, the signal $s_A$ must be sufficiently informative to an extent that the investor at least expects non-negative expected net return of the project after observing $s_A = 1$. That is, the posterior belief updated from $s_A = 1$ must satisfy
	$$Pr(\t_A = 1 | s_A = 1) = \frac{p}{p + (1 - p)(1 - \s_A)} \ge c,$$
	which yields the following necessary condition on $\s_A$ (non-negative profit condition for the investor):		\begin{align}	\label{eq:nc1_sig}
	\s_A \ge \underline{\s} := 1 - (1 - c)\frac{p}{1 - p}.
	\end{align}
	The condition \eqref{eq:nc1_sig} on $\s_A$ implies that PM must reject the promotion of the investment with a sufficiently high probability when the project turns out to an unprofitable one. Note that higher $p$ lowers $\underline{\s}$, which is an intuitive result: If the intrinsic profitability is sufficiently high, the project would be financed even the signal from PM is not that accurate.
	
	Furthermore, PM's signal must dissuade the investor from accessing to her own private information (condition for non endogenous information acquisition). This is because PM will not be able to influence the investor's belief once the investor obtains her own information that accurately predicts the net return of the project. Since $Pr(\t_A = 0 | s_A = 0) = 1$, the investor does not need any additional information if she receives $s_A = 0$. If the investor decides to obtain $e > 0$ after observing PM's signal $s_A = 1$, she will get the expected payoff $Pr(\t_A = 1| s_A = 1) - e.$ If the investor decides not to collect more information by herself, her expected payoff will be $Pr(\t_A = 1 | s_A = 1)(1 - c).$ Therefore, PM can provide the investor with an incentive not to obtain her own information if and only if
	$$Pr(\t_A = 1| s_A = 1) - e \le Pr(\t_A = 1 | s_A = 1)(1 - c),$$
	which yields the following condition:
	\begin{align}	\label{eq:nc2_sig}
	\s_A \ge \overline{\s} := 1 - \frac{e}{c - e} \cdot \frac{p}{1 - p}.
	\end{align}
	The necessary condition \eqref{eq:nc2_sig} requires PM to provide sufficiently accurate information on the project so that the investor has no incentive to access to her own source of the information at some cost. It is also worth noting that $\overline{\s}$ is decreasing in $e$, which is intuitive: the threshold level becomes lower as it becomes more difficult to obtain information (high $e$), which lowers the incentive of PM to be transparent.

 	Lastly, it is necessary to ensure that the information structure assures PM to possess some gains from launching the investment project. Given that PM's signal $s_A = 1$ persuades the investor to fund the project, PM's expected payoff, denoted by $W_A$, is equal to
	\begin{equation}		\label{eq:PM_U_bm}
	\begin{aligned}
	W_A &:= Pr(s_A = 1) \left[ Pr(\t_A = 1 | s_A = 1) - c + \mu_A \right] = Pr(\t_A = 1) + Pr(s_A = 1) (\mu_A - c)	\\
		&= p + [ p + (1 - p) (1 - \s_A) ] (\mu_A - c).
	\end{aligned}
	\end{equation}
	$\mu_A - c$ is the net gains from initiating the investment, which is  positive from \eqref{eq:invest}. Hence, one can easily find that $W_A$ is decreasing in $\s_A$: in order to maximize the chance of funding, PM prefers promoting the project as the one worth investing (i.e. $s_A = 1$), although it is indeed unprofitable ($\t_A = 0$). In other words, PM prefers making $\s_A$ as uninformative as possible to an extent that he can maximize the likelihood of the investment even though the investment gives expected losses to the investor.

	 By summarizing analysis discussed above, we can describe the optimal information policy as Proposition \ref{prop:opt_sig_bm}. Note that there is a threshold for the effort, $\overline{e}$, in the proposition since the investor's incentive to collect additional information is a function of the effort, as equation \eqref{eq:nc2_sig} dictates.
	
\begin{proposition} [Optimal information policy in the benchmark economy]	\label{prop:opt_sig_bm}
There exists a threshold $\overline{e} : = \frac{(1 - c)c}{2 - c}$ such that PM chooses $\s_A = \overline{\s} (> \underline{\s})$ if $e \le \overline{e} $ and $\s_A = \underline{\s} (> \overline{\s})$ if $e > \overline{e}$.
\end{proposition}

	The intuition for the information policy is straightforward. If the cost of acquiring additional information is sufficiently high, the investor will not access to her own source of information although $\s_A$ is relatively low, i.e., PM's signal is not very accurate. As a result, PM can increase the expected surplus $Pr(s_A = 1) \mu_A$ from the investment by lowering $\s_A$, i.e., increasing the probability of $s_A = 1$ conditional on $\t_A = 0$, as much as possible. However, if the cost of collecting the additional information is low, then the investor has a strong incentive to use her own information rather than the public information PM provides. To make his signal reliable, PM must make his signal more creditworthy by increasing $\s_A$. In both cases, the optimal level of accuracy must be equal to the respective lower bounds $\underline{\s}$ and $\overline{\s}$: PM's primary objective is to maximize the possibility of generating the non-financial surplus $\mu_A$, which is equal to $Pr(s_A = 1) = p + (1 - p)(1 - \s_A)$. It is also worth noting from \eqref{eq:nc2_sig} that $\overline{\s}$ is decreasing in $e$: PM's signal must provide more accurate information as the investor has a stronger incentive to access to her own information.

	Above finding has an important but natural implication on the degree of transparency of PM's signal.

\begin{corollary} [Optimal transparency] \label{coro:transparency_benchmark}
The optimal signal becomes more transparent as either $p$ or $e$ is sufficiently low.
\end{corollary}

	Transparency (i.e., high $\s_A$) may not be socially optimal, which coincides with \cite{MS02} but in a different context. The optimal degree of transparency for the policymaking (eg. government, central bank) crucially depends on (1) the nature of the policy and (2) the degree of development of financial/information market in which investors can find (private) information.\footnote{Since we consider a PM that acts as a social planner, the notion of optimality considered in our model is in line with \cite{AngeletosPavan:07}.} If (1) the project can be easily implemented (high $p$) and/or (2) the financial market is underdeveloped, providing less transparent information (see equation \eqref{eq:nc1_sig} and \eqref{eq:nc2_sig}) is optimal. More transparency would be socially beneficial, on the contrary, if (1) the investor holds a pessimistic view on the project (low $p$) and/or (2) the financial market is well-developed. In this case, the investor has a relatively easier access to her own information gathered at the financial market so that she can make a (right) investment decision without relying on the public information. To avoid this unwelcoming decision made by the investor, PM must provide accurate public information to the investor.

\section{Optimal Information Design with Two Projects}	\label{sec:main}

	We now return to the main model with two projects, which yield different values of social surplus. In particular, we first characterize the optimal signal structure and then analyze how it changes with various financial market characteristics, such as the inherent risks of projects and the efficacy of information acquisition by private investors, captured by $p$ and $e$ in our model, respectively.

\subsection{Equilibrium Characterization} \label{sec:eq_cha}

	Let $(\s^{*}_A, \s^{*}_B)$ denote the optimal values of $(\s_A, \s_B)$, which satisfies the following definition of equilibrium.	

\begin{definition} [Equilibrium of two-projects economy]
The equilibrium of the economy with two-projects consists of an information policy, $(\s_A^*, \s_B^*)$, such that

 (1) It is optimal for the investor not to exert an effort for the additional information;

 (2) based on the posterior belief updated after PM's signal, the investor funds one project that gives the highest positive expected payoff;

 (3) the information policy maximizes PM's expected payoff.

\end{definition}	

	There are two possible equilibrium outcomes, either $\s_A \ge \s_B$ or $\s_A < \s_B$. In the first equilibrium, the investor will fund project $A$ whenever she observes $s_A =1$ and fund project $B$ only when she observes $s_A = 0$ and $s_B = 1$ because $Pr(\t_A = 1 | s_A = 1) \ge Pr(\t_B = 1 | s_B = 1)$, implying greater profitability for the project $A$ conditional on the signal structure. In the second equilibrium, the investor will always fund project $B$ after observing $s_B = 1$ but fund project $A$ only after $s_A = 1$ and $s_B = 0$ implying $Pr(\t_A = 1 | s_A = 1) \ge Pr(\t_B = 1 | s_B = 1)$.
	
	Since project $A$ yields a higher social surplus than project $B$ does, PM's key concern is how to maximize the likelihood that the investor funds project $A$. The first outcome corresponds to a transparent investment promotion policy that PM highlights profitability of project $A$ when it is indeed profitable. If such a policy is implemented, PM can successfully persuade the investor to fund project $A$ even after she receives the public information that both projects are promising. However, this promotion policy requires a high level of transparency, reducing the probability that $s_A = 1$ is realized.
	
	In contrast, the second outcome refers to an alternative investment promotion policy that PM overly advertises profitability of investing in project $A$ although the investor is likely to lose money --- i.e., $\t_A = 0$. If PM adopts this promotion strategy, PM can induce the investment in project $A$ only if the investor receives bad news about project $B$. Nevertheless, such a promotion strategy can be optimal since PM expects a high likelihood that the investor receives good news on project $A$ although it is unprofitable.
	
	Therefore, a key focus of our analysis is how PM's optimal investment promotion strategy varies with the parameters of financial market characteristics $p$ and $e$. To study how the equilibrium outcome is determined, first suppose that PM's optimal signal is $\s^*_A \ge \s^*_B$. In equilibrium, the investor will fund project $B$ only when $s_A = 0$ is drawn. Hence, the signal $s_B$ must maximize $p + [ p + (1 - p) (1 - \s_B) ] (\mu_B - c)$, the expected payoff from investing in project $B$ conditional on $s_B = 1$. From Proposition \ref{prop:opt_sig_bm}, we have
	\begin{align}	\label{eq:opt_sb_c1}
	\s^{*}_B = \hat{\s}^* := (\underline{\s} \vee \overline{\s}),
	\end{align}
	where $\underline{\s} < \overline{\s}$ if and only if $e \ge \overline{e}$.
	
	However, the investor funds project $A$ whenever $s_A = 1$ is realized. Therefore, the signal $s_A$ must optimize the following ex-ante expected payoff
	$$Pr(s_A = 1) \left[ Pr(\t_A = 1 | s_A = 1) + \mu_A - c \right] + Pr(s_A = 0) Pr(s_B = 1) \left[ Pr(\t_B = 1 | s_B = 1) + \mu_B - c \right],$$
	which is rewritten as
	\begin{align}	\label{eq:payoff_c1}
	U_A(\s_A) := (p + \mu_A - c) + (1 - p) \s_A  \left[ \left\{p + p \left( (2 - c) \wedge \frac{c}{c - e} \right) ( \mu_B - c) \right\} - (\mu_A - c) \right].	
	\end{align}
	We throughout denote $U_A$ as the expected payoff when the optimal signal structure is $\s^*_A \ge \s^*_B$.

	PM's expected payoff is decomposed into two separate terms. The first term, $p + \mu_A - c$ is simply the expected total surplus from the investment in project $A$. More importantly, the second term accounts for the net gains from sending the bad signal about project $A$ ($s_A = 0$). First, the terms inside the curly brackets represent the ex-ante expected social surplus from $s_A = 0$, which equals the probability of $s_B = 1$ multiplied by the expected total surplus from the investment in project $B$. Second, the term $(\mu_A - c)$ is the net social surplus from funding project $A$. In sum, the terms inside the square brackets, which is defined as
	\begin{align} \label{eq:La}
	L_A (p,e) := \left\{p + p \left( (2 - c) \wedge \frac{c}{c - e} \right) ( \mu_B - c) \right\} - (\mu_A - c),
	\end{align}
	are the marginal gains from increasing accuracy of the signal $s_A$: as the probability of $s_A = 0$ increases, PM has to forgo the chances of enjoying the social surplus $\mu_A - c$ that project $A$ generates, whereas it will be more likely to get the surplus from project $B$ that is realized after the investor receives $s_B = 1$. The following lemma would be useful for the subsequent analysis.

\begin{lemma} \label{lemma:L_A}
The equilibrium with $\sigma^*_A \ge \s^*_B$ exists only if
 \begin{align}	\label{eq:nc_c1}
 L_A(p,e) > 0.
 \end{align}
\end{lemma}	
	
	In fact, if PM prefers $\s_A \ge \s^*_B$ and the condition \eqref{eq:nc_c1} is satisfied, the optimal level of accuracy of the signal $s_A$ must be $\s^*_A = 1$. If not, it will be too risky for PM to make the signal $s_A$ more informative than that with the accuracy $\s^*_A = \hat{\s}^*$.
	
	Next, suppose the optimal policy $\s^*_B > \s^*_A$. In the corresponding equilibrium, the investor funds project $B$ whenever $s_B = 1$, while she funds project $A$ only after receiving $s_A = 1$ but $s_B = 0$. This outcome is symmetric to the former outcome with $\s^*_A \ge \s^*_B$. Hence, it is straightforward to observe
	\begin{align}	\label{eq:opt_sa_c2}
	\s^{*}_A = \hat{\s}^*.
	\end{align}
	One can also easily find that PM's total ex-ante expected payoff, denoted by $U_B$, is
	\begin{align}	\label{eq:payoff_c2}
	U_B(\s_B) := (p + \mu_B - c) + (1 - p) \s_B  \left[ \left\{p + p \left( (2 - c) \wedge \frac{c}{c - e} \right) ( \mu_A - c) \right\} - (\mu_B - c) \right].	
	\end{align}
	Like $L_A (p, e)$, we can denote the terms inside the square brackets by
	\begin{align} \label{eq:Lb}	
	L_B(p,e) \equiv  \left\{p + p \left( (2 - c) \wedge \frac{c}{c - e} \right) ( \mu_A - c) \right\} - (\mu_B - c).
	\end{align}
	Then we have the following observation.

\begin{lemma} \label{lemma:L_B}
The equilibrium with $\s^*_B > \s^*_A$ exists only if $L_B(p,e)>0$.
\begin{align}	\label{eq:nc_c2}
L_B (p, e) > 0.
\end{align}
\end{lemma}
	
	Once again, if PM prefers $\s_B > \s^*_A$ and the condition \eqref{eq:nc_c2} is satisfied, the optimal accuracy of the signal $s_B$ must be $\s^*_B = 1$. If not, then PM refrains from making $s_B$ more informative than the one with $\s_B = \hat{\s}^*$, and therefore, the optimal signal cannot be the ones with $\s^*_B > \s^*_A$.
	
	From the discussion above, one can find that a key determinant of the optimal signal structure $(\s^*_A, \s^*_B)$ is the signs of $L_A(p, e)$ and $L_B(p, e)$. If $L_A\le 0$ and $L_B\le 0$, PM has no incentive to make any of either $\s_A$ or $\s_B$ more accurate than the minimum level of accuracy $\hat{\s}^*$ that barely induces the investors to fund project $i \in \set{A, B}$ without any additional information. If $L_A > 0 \ge L_B$, PM prefers $\s^*_A = 1$ and $\s^*_B = \s^*$: improving informativeness of the signal $s_A$ always increases PM's expected payoff, whereas it is too risky to increase the accuracy of $s_B$ from $\s^*$. If $L_B > 0 \ge L_A$, on the contrary, PM's optimal (and only feasible) strategy is $\s^*_A  = \hat{\s}^*$ and $\s^*_B = 1$. If $L_A > 0$ and $L_B > 0$, both strategies, $(\s^*_A , \s^*_B) = (1, \hat{\s}^*)$ and $(\s^*_A, \s^*_B) = (\hat{\s}^*, 1)$, are feasible. In this case, PM compares two options and chooses the one that gives a higher expected payoff than the other.
	
	A key parameter that determines the signs of $L_A$ and $L_B$ is $p$, where $1 - p$ captures the inherent risk of each project $i \in \set{A, B}$. First,  it follows from $\mu_A - c > \mu_B - c > 0$ that $L_A < L_B$ for every $p$. Furthermore, from Lemma \ref{lemma:L_A} and \ref{lemma:L_B}, one can easily observe that both $L_A$ and $L_B$ are linear and increasing in $p$. Specifically, $L_B > 0$ if and only if
	\begin{align}	\label{eq:lb>0}
	p > \underline{p}^* := \frac{(\mu_B - c)}{1 + \frac{c}{c - e} (\mu_A - c)},
	\end{align}
	and $L_A > 0$ if and only if
	\begin{align}	\label{eq:la>0}
	p > \tilde{p}^* := \frac{(\mu_A - c)}{1 + \frac{c}{c - e} (\mu_B - c)}.
	\end{align}

	Since $\mu_A > \mu_B > c$, we have $\tilde{p}^* > \underline{p}^*$. From equations \eqref{eq:lb>0} and \eqref{eq:la>0}, it can be observed that $(\s^*_A, \s^*_B) = (\hat{\s}^*, 1)$ can be a feasible information structure of the signal $(s_A, s_B)$ if and only if $p \ge \underline{p}^*$, and $(\s^*_A, \s^*_B) = (1, \hat{\s}^*)$ can be feasible if and only if $p > \tilde{p}^*$. Lastly, for any given $p \ge \tilde{p}^*$, PM will choose $(\s^*_A, \s^*_B) = (1, \hat{\s}^*)$ if and only if
	$$U_A(1) - U_B(1) = p (\mu_A - \mu_B) \left[ 1 - (1 - p) \left( (2 - c) \wedge \frac{c}{c - e} \right) \right] \ge 0,$$
	which is equivalent to
	\begin{align}	\label{eq:ua>ub_sim}
	p \ge \hat{p}^* := \left( \frac{1 - c}{2 - c} \right) \wedge \left( \frac{e}{c} \right).
	\end{align}

	The optimality of $(\s^*_A, \s^*_B) = (1, \hat{\s}^*)$ obtains when conditions \eqref{eq:la>0} and \eqref{eq:ua>ub_sim} are jointly satisfied. Let $\overline{p}^*$ denote the maximum of $\tilde{p}^*$ and $\hat{p}^*$, i.e., $\overline{p}^* := \tilde{p}^* \vee \hat{p}^*$. Then we can derive the optimal information policy as is described in Proposition \ref{prop:main_result}.
	
\begin{proposition}	[Optimal information policy in the two-projects economy]\label{prop:main_result}
There exist two cutoffs $0 \le \underline{p}^* \le \overline{p}^* \le c$ such that
\begin{equation}	 \label{eq:opt_sig}
\begin{aligned}
(\s^*_A, \s^*_B) = \left \{
\begin{array}{l l}
(\hat{\s}^*, \hat{\s}^*) &\; \mbox{   if    } \;\; p < \underline{p}^*, \\
(\hat{\s}^*, 1) &\; \mbox{   if    } \;\; p \in [\underline{p}^*, \overline{p}^*), \\
(1, \hat{\s}^*) &\; \mbox{   if    } \;\; p \ge \overline{p}^*, \\
\end{array}
\right.
\end{aligned}
\end{equation}
\end{proposition}
	
	Proposition \ref{prop:main_result} tells us how the risk of providing accurate information influences PM's optimal design of the experiment for the promotion of the investments. In other words, the inherent risk of the financial projects determines the adjusted risks of projects exposed to PM, which in turn influences the informational structure of the optimal signals. To see why, first suppose that $p < \underline{p}^*$ for which we have $L_A < L_B \le 0$. In this case, PM suffers net expected losses from failing to promote the investment in any project, i.e., $s_i = 0$ for each $i \in \set{A, B}$. To avoid these losses for each $i \in \set{A, B}$, PM minimizes the expected losses incurred by $s_i = 0$ by lowering the probability $Pr(s_i = 0) = (1 - p) \s_i$ with $\s_i = \hat{\s}^*$.
	
	Second, fix any $p  \in [\underline{p}^*, \tilde{p}^*)$ that yields $L_A \le 0 < L_B$. In this case, PM can get a positive expected surplus conditional on $s_B = 0$. Specifically, a high $p$ implies a sufficiently high probability of $s_A = 1$, thus PM expects expected gains from the investment in project $A$ although he fails to promote investment in project $A$. However, PM will suffer net expected losses from failing to promote investment in project $A$ whenever $s_A = 0$ is realized. Therefore, PM maximizes the expected gains from promoting investment in project $A$ by setting $\s^*_B = 0$ that maximizes the probability $Pr(s_B = 0)$, while he minimizes the expected losses from failure of investment in project $A$ by setting $\s^*_A = \hat{\s}^*$ that minimizes the probability $Pr(s_A = 0)$.
	
	Third, consider $p \ge \tilde{p}^*$ that yields $0 \le L_A < L_B$. In this case, there is no risk of failing to promote the investment in any project $i \in \set{A, B}$ after $s_i = 0$. Since $\mu_A > \mu_B > c$, PM's main purpose of the investment promotion policy is then the maximization of the probability that the investor funds project $A$. To accomplish this mission, PM has two choices that have a tradeoff. First, PM can maximize the probability $Pr(s_A = 1) = p + (1 - p)(1 - \s_A)$ by decreasing $\s_A$ to $\s_A = \hat{\s}^*$. However, this promotion strategy makes the signal $s_A$ relatively uninformative, so it can only induce the investor to fund project $A$ only when she receives $s_B = 0$ and has the negative posterior belief about project $B$. Second, PM can highlight profitability of project $A$ when its true value is $\t_A = 1$ by increasing $\s_A = 1$. This strategy makes the investor unconditionally fund project once she observes $s_A = 1$. However, the probability of $s_A = 1$ inevitably decreases, so the overall probability of funding project $A$ is possibly lower than the first option.
	
	We find that $p$, the inherent risk of the financial projects, once again determines the optimal policy for promoting investment in project $A$. First, fix a relatively low value of $p \in [\tilde{p}^*, \tilde{p}^* \vee \hat{p}^*)$. In this case, project $B$ is likely to be unprofitable ($\t_B = 0$), and therefore, $s_B = 0$ is likely to be drawn relative to $s_B = 1$ for any given $\s_B \in [0, 1]$. To promote investment in project $A$ effectively, it is important for PM to induce the investor to fund project $A$ unconditional on $s_B = 0$. Hence, PM must increase the probability of $s_A = 1$ because the investor, regardless of informativeness of the signal $s_A$, will choose project $A$ after receiving $s_A = 1$ and $s_B = 0$. Correspondingly, PM has to maximize the probability that the investor receives $s_B = 0$ rather than $s_B = 1$ by setting $\s^*_B = 1$. Next, suppose a sufficiently high $p > \tilde{p}^* \vee \hat{p}^*$. In this case, $s_B = 1$ will be likely to be drawn compared to $s_B = 0$. Since the investor is likely to have favorable posterior belief about project $B$, PM must optimally underscore profitability of investing in project $A$ after $s_A = 1$ by making the signal $s_A$ fully informative. In accordance, PM must set $\s^*_B = \hat{\s}^*$: PM's promotion strategy is to highlight the profitability of project $A$ relative to project $B$ when the investor holds favorable posterior belief about both projects.

In sum, the information policy has an interesting implication on the relationship between the intrinsic profitability and the degree of transparency, which is summarized by the following corollary.

\begin{corollary} [Transparency of the investment promotion policy] \label{coro:transparency}
 In the economy with $\mu_A>\mu_B$, it is optimal for PM to be perfectly transparent about the profitability of project A only when the profitability is substantially high ($p\geq \bar{p}^*$). On the contrary, to perfectly reveal its own information is optimal for project B only when its profitability is at the middle ($ p \in [\underline{p}^*, \overline{p}^*)$).

\end{corollary}
		
	Proposition \ref{prop:main_result} also provides an interesting testable prediction: governments may drastically change their investment promotion policies with the confidence of the private investors. First, suppose that market participants are strongly pessimistic about investing in new industries a government targets to develop. Then the government adopts a promotion policy that \emph{exaggerates} the profitability of investing in those industries on average by pooling some incompetent firms in the industries with competent ones. Second, suppose that the outside investors are relatively optimistic about funding new industries that the government plans to develop. In this case, the government promotes investment in these new industries by providing fully accurate information on the profitability so that the private investors can \emph{separate the wheat from the chaff}. In sum, our analysis can be empirically verified if the firms that earn the government promotions have lower (\textit{resp.} higher) default rates on their debts when the financial market is bullish (\textit{resp.} bearish).

\subsection{Comparative Statics} \label{sec:comparative}

 One of the core assumptions of our model is that investment projects are heterogenous in their effects on the aggregate economy, captured by $\mu_A>\mu_B$. It is natural, hence, to study how the equilibrium property alters when the difference in such an effect becomes greater. For simplicity of the analysis, assume that $\mu_A$ becomes higher while $\mu_B$ is fixed. From equation \eqref{eq:lb>0} and \eqref{eq:la>0}, it is easy to observe that higher $\mu_A$ lowers $\underline{p}^*$ while increases $\tilde{p}^*$, which implies the following proposition.

 \begin{proposition} [Comparative statics with $\mu_A$] \label{prop:mu_A}
 Suppose that $\tilde{p}^* > \hat{p}^*$. If $\mu_A$ becomes higher, (1) the region of $p$ at which PM does not fully reveal its information on project A enlarges while (2) the region of $p$ at which PM perfectly reveals its information on project B enlarges.
 \end{proposition}

  This implies that it is optimal for PM to be less transparent on project A while to reveal more information on project B as differences in externality becomes greater, which results in higher (\textit{resp.} lower) ex-post expectation on profitability of project A (\textit{resp.} project B).
		
	We complete our analysis by studying how the optimal promotion policy is associated with  $e > 0$, the cost of information acquisition by the investor. Specifically, we conduct comparative statics analysis of how the cutoff values $\underline{p}^*$ and $\overline{p}^*$ vary with the investor's information acquisition cost $e$. To this end, we restrict our attention to the cases $e \le \bar{e}$. In fact, one can easily find the following results by differentiating the respective thresholds $\underline{p}^*, \tilde{p}^*$, and $\hat{p}^*$ in \eqref{eq:lb>0} -- \eqref{eq:ua>ub_sim} by $e$.
	
\begin{proposition} [Comparative statics with $e$]	\label{prop:cs_e}
For every $e \le \overline{e}$, we have the following:
\begin{enumerate}[(i)]

	\item $\underline{p}^*$ is decreasing in $e$;
	
	\item there exists a $\hat{e}^* \le \overline{e}$ such that $\overline{p}^*$ is decreasing in $e$ if $e \le \hat{e}^*$, but $\overline{p}^*$ is increasing in $e$ otherwise.

\end{enumerate}
\end{proposition}
	
	From Lemma \ref{lemma:L_B}, Proposition \ref{prop:cs_e}-(i) is immediate since $L_B$ is increasing in $e$. As $e$ increases, the investor becomes less willing to acquire her own information and thus relies more on PM's information instead. Thus, PM can enjoy higher expected payoff from $s_B = 0$ by increasing $Pr(s_A = 1 | \t_A = 0)$ because the investor will fund project $A$ after $s_A = 1$ although this realized signal is less informative. Hence, the signal with $(\s^*_A, \s^*_B) = (\hat{\s}^*, 1)$ in turn becomes PM's most preferred option even for lower values of $p$ as $e$ increases, thereby leading to $\frac{d \underline{p}^*}{d e} < 0$.
	
	However, it is not obvious whether $\overline{p}^*$ is increasing or decreasing in $e$. Indeed, there are two necessary (and sufficient if combined altogether) conditions for the optimality of $(\s^*_A, \s^*_B) = (1, \hat{\s}^*)$. First, $L_A$ should not be negative. That is, it should be highly likely that project $B$ is funded, and thus PM must not suffer losses from the failure of promoting the investment in project $A$ after $s_A = 0$. This condition ensures feasibility of the signal structure $(\s^*_A, \s^*_B) = (1, \hat{\s}^*)$. Furthermore, this condition is satisfied if and only if $p \ge \tilde{p}^*$, as was seen in \eqref{eq:la>0}. Second, the signal structure $(\s_A, \s_B) = (1, \hat{\s}^*)$ must give PM a higher expected payoff than that with $(\s_A, \s_B) = (\hat{\s}^*, 1)$. This condition ensures optimality of $(\s^*_A, \s^*_B) = (1, \hat{\s}^*)$, and is satisfied if and only if $p \ge \hat{p}^*$, as was shown in \eqref{eq:ua>ub_sim}. The changes of $\overline{p}^*$ with $e$ therefore crucially depends on which of these two conditions, \eqref{eq:la>0} and \eqref{eq:ua>ub_sim}, is stronger than the other. If the condition \eqref{eq:la>0} is stronger, then $\overline{p}^* = \tilde{p}^*$, and thus $\overline{p}^*$ decreases in $e$. On the other hand, if the other condition \eqref{eq:ua>ub_sim} is stronger, then $\overline{p}^* = \hat{p}^*$, and thus $\overline{p}^*$ increases in $e$.
	
	Proposition \ref{prop:cs_e}-(ii) shows that the necessary condition $L_A \ge 0$ in \eqref{eq:la>0} is stronger than the other condition $U_A(1) \ge U_B(1)$ in \eqref{eq:ua>ub_sim} if $e$ is relatively low, and vice versa. When $e$ is sufficiently small, PM must send accurate information to the investor in order to keep her from gathering her own information. Thus, the probability of funding project $i$ conditional on $\t_i = 0$ will diminish. Particularly, this negative effect will hit PM more severely when he chooses $(\s_A, \s_B) = (\hat{\s}^*, 1)$: since $Pr(s_A = 1) = p + (1 - p) (1 - \hat{\s}^*)$, a small $e$ lowers the probability of funding project $A$ that gives large surplus $\mu_A$; since $Pr(s_B = 1) = p$, there is no effect of $e$ on the probability of funding project $B$ that yields small surplus $\mu_B$. Hence, PM's best choice will be $(\s_A, \s_B) = (1, \hat{\s}^*)$ once the associated feasibility condition \eqref{eq:la>0} is satisfied. %One can also easily observe from \eqref{eq:la>0} that $\overline{p}^*$ is decreasing in $e$: PM, for a fixed $p$, can extract more surplus from the investor by promoting investment in project $A$ although it is not profitable.
	
	Next, consider the opposite case where $e$ is relatively high. In this case, the investor refrains from using her own information because it is too costly to gather. This means that PM's signals need not to be strongly accurate to effectively promote investments in the financial projects. Therefore, the necessary conditions $L_A \ge 0$ and $L_B \ge 0$ are satisfied even for sufficiently small values of $p$: that is, a high $e$ weakens the necessary condition \eqref{eq:la>0}. On the other hand, a high effort cost of information acquisition make PM prefer the signal structure $(\s_A, \s_B) = (\hat{\s}^*, 1)$ to $(\s_A, \s_B) = (\hat{\s}^*, 1)$: a high $e$ increases the probability $Pr(s_A = 1) = p + (1 - p) (1 - \hat{\s}^*)$ of funding project $A$, which gives larger non-financial surplus $\mu_A$ than project $B$. Hence, for high values of $e$, PM will not choose $(\s_A, \s_B) = (1, \hat{\s}^*)$ unless $p$ is high enough to guarantee $U_A(1) \ge U_B(1)$.
	
	It is worth noting that the optimal promotion policy can be dramatically shifted as the financial markets provide useful information to the market participants more effectively. However, these effects of the financial market development on the optimal investment promotion policy depends on the ex-ante risks of financial projects. For instance, suppose $p$ is low to an extent that we can have either $p \in [\underline{p}^*, \overline{p}^*)$ or $p < \underline{p}^*$. If $e$ is relatively high, then we have $p \in [\underline{p}^*, \overline{p}^*)$. In this case, the optimal promotion strategy is to send relatively uninformative signal about project $A$ ($\s^*_A = \hat{\s}^*$) but perfectly informative signal about project $B$ ($\s^*_B = 1$). However, if $e$ is sufficiently low to an extent that $p < \underline{p}^*$, PM becomes strongly concerned about failing to raise capital for the projects. Thus, the optimal promotion strategy will be shifted to exaggerating profitability of both projects although they are not ($\s^*_A = \s^*_B = \hat{\s}^*$). In sum, when the investment in a financial project is viewed as being inherently risky, PM refrains from sending fully informative signal about each project as the investor can obtain her own information at the financial markets in a more cost-effective way.
	
	If $p$ is relatively high, there may be no monotonic relationship between the information that PM's signals contain and effectiveness of the information acquisition at the financial markets. This result arises from the observation that $\overline{p}^*$ is convex in $e$. If $e$ is relatively high so that $p < \overline{p}^*$, PM prefers sending fully informative signal about project $B$ but relatively uninformative signal about project $A$ (i.e., $(\s^*_A, \s^*_B) = (\hat{\s}^*, 1)$): PM believes that $p$ is not high enough to promote the investment in project $A$ by truthfully highlighting its profitability regarding that the investor are likely to receive $s_A = s_B = 1$. As $e$ decreases to an extent $p \ge \overline{p}^*$, PM changes his view on $p$, and also changes the promotion policy to the signals with $(\s^*_A, \s^*_B) = (1, \hat{\s}^*)$. However, if $e$ decreases further so that we have $p < \overline{p}^*$, PM perceives that he will make losses on average if $s_A = 0$ is drawn. To minimize these losses from failing to fund project $A$, PM returns his promotion policy once again to $(\s^*_A, \s^*_B) = (\hat{\s}^*, 1)$ so that he can minimize the probability that the investor chooses project $B$.

\section{Conclusion} \label{sec:conclusion}

In this paper, we apply the Bayesian persuasion framework, \emph{\`{a} la} \citet{KG2011}, to an economy with a government that tries to promote large-scale investment projects with positive externality under information frictions and derive the optimal information strategy. We then investigate how the optimal policy varies when the underlying features of the economy change.

The most important distinctive feature of our analysis from the previous literature on an industrial policy (including state-promoted large scale investment) is that we emphasize the informational role of the government to fund the projects in the presence of \textit{(i)} information friction and \textit{(ii)} ignorance of private investors on positive spillover effects from the projects. While the model analyzed in our work is simple hence parsimonious, we believe that the framework can be further utilized in various economic circumstances. For instance, environmental policy is one of the very important projects that cannot be successfully implemented without efforts by policy makers but has positive externality on the aggregate economy \citep{Gates2021}. One can easily apply our model to the economy with a policy maker who wants to maximize its own utility, which may harm the economy (eg. negative spillover effect). In addition, as was discussed at the end of Section \ref{sec:eq_cha}, our findings have testable implications with the data; utilizing the data on firm's default rates, one would be able to test if government really has incentives to promote some particular industries.

One potential caveat of our model is that we do not consider the role of endogenous public signals \citep{RS14jet, BG2015} in shaping equilibrium outcomes. Another interesting extension of our model is to consider reputation concerns of a policy maker and/or learning by private investors from equilibrium outcomes realized in the previous periods in a dynamic setup \citep{AHP:07}, which might affect our findings. We leave such interesting extensions as future works.

% Given that our analysis can also be applied to a situation in which the policy makers have objective functions not aligned to social welfare (eg. negative spillover effect), we believe that the framework provided in this paper can be used (and extended) in various economic circumstances.

%While our model has several assumptions to make , and is silent on a number of issues that are beyond the scope of the paper. It is our hope, however, that our findings lay the groundwork for studying these and other challenging issues in future research.

\bk 	
	
\pagebreak
	
\bibliographystyle{aea}
\bibliography{bayesian_bib}

\newpage

\appendix

\section*{Appendix}

\section{Proofs} \label{app_proof}

\begin{proof}[Proof of Proposition \ref{prop:opt_sig_bm}]

	From \eqref{eq:PM_U_bm}, PM's expected payoff is decreasing in $\s_A$. Therefore, it is optimal for PM to lower $\s_A$ as much as possible subject to the constraints that $P(\t_A = 1 | s_A = 1) \ge c$ and the investor gets a higher expected payoff from not acquiring her own information at the cost $e$, which implies $\s_A = (\overline{\s} \vee \underline{\s})$.	\end{proof}

\begin{proof}[Proof of Lemma \ref{lemma:L_A}]
From \eqref{eq:payoff_c1}, it can be easily observed that $U_A$ is increasing in $\s_A$ if and only if $L_A \ge 0$. Hence, for $\s_A = \s^*_A (\ge \s^*_B = \hat{\s}^* > 0)$ to be PM's optimal strategy, we must have $L_A \ge 0$. \end{proof}

\begin{proof}[Proof of Lemma \ref{lemma:L_B}]
One can prove by applying the same logic used for the proof of Lemma \ref{lemma:L_A}, so we omit the formal proof.\end{proof}

\begin{proof}[Proof of Proposition \ref{prop:main_result}]

	Without loss of generality, we restrict our attention to the cases $\underline{p}^* < \tilde{p}^* < \hat{p}^*$.

	To prove part (i), fix any $p < \underline{p}^*$ at which $L_A, L_B < 0$. We first show that the optimal signal must yield $Pr(\t_A = 1 | s_A = 1) \ge Pr(\t_B = 1 | s_B = 1)$. To show this by contradiction, suppose $Pr(\t_A = 1 | s_A = 1) < Pr(\t_B = 1 | s_B = 1)$, which implies $\s^*_A < \s^*_B$. Then we must have $\s^*_A = \hat{\s}^* > 0$ and PM's total expected payoff must be equal to $U_B$. However, since $L_B < 0$ for any $p < \underline{p}^*$, PM can increase his expected payoff by lowering $\s_B$ below $\s^*_B$, a contradiction. Furthermore, since the investor funds project $B$ only after receiving $s_A = 0$ and $s_B = 1$, PM maximizes his expected payoff conditional on $s_A = 0$ by setting $\s^*_B = \hat{\s}^*$. Lastly, PM's total expected payoff from $\s_A \ge \s_B$ is $U_A$, which is decreasing in $\s_A$ since $L_A < 0$. Therefore, the optimal signal is $\s^*_A = \s^*_B = \hat{\s}^*$.
	
	To prove part (ii), fix any $p \in [\underline{p}^*, \tilde{p}^*)$, where $L_A < 0 \le L_B$. Since $L_A < 0$, an optimal signal must have the structure $\s^*_A < \s^*_B$. Moreover, since the investor always funds project $B$ once she receives $s_B = 1$, PM maximizes the expected payoff conditional on $s_B = 0$ by setting $\s^*_A = \hat{\s}^*$. Lastly, PM's total expected payoff $U_B$ is increasing in $\s_B$, which implies $\s^*_B = 1$.
	
	Next, fix any $p \in  [\tilde{p}^*, \hat{p}^*)$. Since $p \ge \tilde{p}^*$, we have $L_B > L_A \ge 0$, and therefore, PM's optimal signal must be either $(\s^*_A, \s^*_B) = (1, \hat{\s}^*)$ or $(\s^*_A, \s^*_B) = (\hat{\s}^*, 1)$. Therefore, PM chooses one of these two signals that yields higher expected payoff than the other. To find which of two signals yields higher expected payoff, we rewrite $U_A$ as
	$$U_A = K(p, c) + p \left[ \mu_A + (1 - p) (2 - c) \wedge \frac{c}{c - e} \mu_B \right]$$
	and $U_B$ as
	$$U_B = K(p, c) + p \left[ \mu_B + (1 - p) (2 - c) \wedge \frac{c}{c - e} \mu_A \right],$$	
	where $K(p, c)$ is the residual polynomial of $p$ and $c$ that excludes $\mu_A$ and $\mu_B$. Since $\mu_A > \mu_B$, $U_A - U_B$ is increasing in $p$. Furthermore, it follows from the definition of $\hat{p}^*$ that $U_B > U_A$ for all $p \in  [\tilde{p}^*, \hat{p}^*)$, which implies the optimal signal is $(\s^*_A, \s^*_B) = (\hat{\s}^*, 1)$.
	
	To prove part (iii), fix any $p \ge \hat{p}^*$. Since $L_B \ge L_A > 0$ and $U_A \ge U_B$, the optimal signal is $(\s^*_A, \s^*_B) = (1, \hat{\s}^*)$.	\end{proof}

\begin{proof}[Proof of Proposition \ref{prop:mu_A}]

It is immediate from equation \eqref{eq:lb>0} that higher $\mu_A$ lowers $\underline{p}^*$ while it increases $\tilde{p}^*$ from equation \eqref{eq:la>0}. Hence, it follows that the region at which PM lies about project A ($p<\tilde{p}^*$) enlarges while PM is truthful abut project B ($ p \in [\underline{p}^*, \tilde{p}^*)$) become larger with $\tilde{p}^* > \hat{p}^*$.

\end{proof}

\begin{proof}[Proof of Proposition \ref{prop:cs_e}]

	We first prove part (i). By definition of $\underline{p}^*$ as stated in \eqref{eq:lb>0}, it is obvious that $\underline{p}^*$ is decreasing in $e$.
	
	We next prove part (ii). From \eqref{eq:la>0}, $\frac{d\tilde{p}^*}{de} < 0$ is immediate. Furthermore, From \eqref{eq:ua>ub_sim}, one can easily observe $\frac{d \hat{p}^*}{de} > 0$. Since $\lim_{e \rightarrow 0} \tilde{p}^* > 0 = \lim_{e \rightarrow 0} \hat{p}^*$, there exists a $\hat{e}^* \le \overline{e}^*$ such that $\overline{p}^* = \tilde{p}^*$ if $e \le \hat{e}^*$ and $\overline{p}^* = \hat{p}^*$ otherwise, which also implies $\overline{p}^*$ is decreasing in $e$ if $e \le \hat{e}^*$ and increasing in $e$ otherwise. 	\end{proof}

\end{document}